%% file: main.tex
\documentclass[a4paper]{article}

\usepackage[utf8]{inputenc}
\usepackage[a4paper, margin=.85in]{geometry}
\usepackage{marginnote}

\usepackage[lined, noend, linesnumbered, boxed]{algorithm2e}
\SetKwBlock{Init}{}{}

\let\oldnl\nl% Store \nl in \oldnl
\newcommand{\nonl}{\renewcommand{\nl}{\let\nl\oldnl}}% Remove line number for one line

\SetAlCapSkip{5pt}

\usepackage{mathtools}
\usepackage{amsthm}
\usepackage{amsmath}

\newcommand{\qH}[1]{\qedhere~$_{\text{#1}}$}
\newenvironment{proofsketch}{%
  \renewcommand{\proofname}{Proof Sketch}\proof}{\endproof}

\usepackage{cleveref}
\newtheorem{theorem}{Theorem}
\newtheorem{assumption}[theorem]{Assumption}
\newtheorem{definition}[theorem]{Definition}
\newtheorem{observation}[theorem]{Observation}
\newtheorem{corollary}[theorem]{Corollary}
\newtheorem{lemma}[theorem]{Lemma}
\newtheorem*{claim*}{Claim}

\newtheorem{claiminner}{Claim}
\newenvironment{claim}[1]{
    \renewcommand\theclaiminner{#1}
    \claiminner
}{\endclaiminner}

\def\indentation{2em}

\makeatletter
\newtheoremstyle{indented}
  {.5em}% space before
  {.5em}% space after
  {\addtolength{\@totalleftmargin}{\indentation}
   \addtolength{\linewidth}{-\indentation}
   \parshape 1 \indentation \linewidth}% body font
  {}% indent
  {\bfseries}% header font
  {.}% punctuation
  {.5em}% after theorem header
  {}% header specification (empty for default)
\makeatother

\theoremstyle{indented}
\newtheorem{claimindentedinner}{Claim}
\newenvironment{claimindented}[1]{
    \renewcommand\theclaimindentedinner{#1}
    \claimindentedinner
}{\endclaimindentedinner}

\makeatletter
\newenvironment{proofindented}[1][\proofname]{\par
  \pushQED{\qed}%
  \normalfont \topsep6\p@\@plus6\p@\relax
  \list{}{\leftmargin=\indentation
          \rightmargin=0em
          \settowidth{\itemindent}{\itshape#1}%
          \labelwidth=\itemindent
          \parsep=0pt \listparindent=\parindent 
  }
  \item[\hskip\labelsep
        \itshape
    #1\@addpunct{.}]\ignorespaces
}{%
  \popQED\endlist\@endpefalse
}
\makeatother

\usepackage{tikz}
\newcommand*\circled[1]{\tikz[baseline=(char.base)]{
            \node[shape=circle,draw,inner sep=2pt, solid] (char) {#1};}}

\newcommand{\CAS}{compare-and-swap}
\newcommand{\GCAS}{generalized-compare-and-swap}
\newcommand{\FA}{fetch-and-increment}
\newcommand{\clockobject}{C}
\newcommand{\announceobject}{A}
\newcommand{\stateobject}{S}
\newcommand{\helpobject}[1]{R_{#1}}
\newcommand{\helpobjectsymbol}{h}
\newcommand{\old}{val_1}
\newcommand{\new}{val_2}

\title{Generalized Compare and Swap}

\author{Vassos Hadzilacos \and Myles Thiessen \and Sam Toueg}

\begin{document}

\maketitle

\input{sections/abstract}

\input{sections/introduction}

\input{sections/model}

\input{sections/universal_construction}

\input{sections/definitions_and_basic_facts}

\input{sections/proof_of_wait_freedom}

\input{sections/proof_of_linearizability}

\input{sections/conclusion}

\bibliography{main}

\end{document}

%% file: sections/abstract.tex
\begin{abstract}
    In this paper, we first propose a natural generalization of the well-known \CAS{} object, one that replaces the equality comparison with an arbitrary comparator.
    We then present a simple wait-free universal construction using this object and prove its correctness.
\end{abstract}

%% file: sections/introduction.tex
\section{Introduction}
\label{sec:introduction}

\emph{Compare-and-swap} (CAS) is a simple base object, often supported in hardware, that is commonly used in shared memory systems to implement other objects that are not supported in hardware~\cite{fatourou2009redblue, herlihy1991wait, fatourou2011highly, anderson1995universal, chuong2010universal, afek1995wait, naderibeni2023wait, michael1996simple}.
Each CAS object $O$ stores a single value and supports the following two operations: \emph{Read($O$)} which returns the current value of $O$, and \emph{CAS($O$, $\old$, $\new$)} which compares the current value of $O$ to $\old$ and, \mbox{if they are \emph{equal},} replaces the current value of $O$ with $\new$ as shown below.

\input{figures/cas}

In this paper, we propose a natural generalization of CAS called \mbox{\emph{\GCAS{}}~(GCAS)}: GCAS replaces the equality comparison of CAS with an arbitrary comparator as shown below.

\input{figures/gcas}

Note that while the above definition of GCAS supports an arbitrary comparator, a restricted version of GCAS that only supports the usual comparators ($=$, $>$, $<$, $\geq$, $\leq$) may be easier to support in hardware: the cost of doing an inequality comparison in GCAS should be similar to the cost of doing an equality comparison in CAS (which is currently supported in hardware).

To illustrate the use of GCAS, we present a simple \emph{wait-free universal construction}~\cite{herlihy1991wait} from GCAS and \emph{\FA{}} (F\&I) objects; this universal construction works in a system with an unbounded number of asynchronous processes.
In a nutshell, the \FA{} object is used to timestamp operations, and the GCAS uses the $>$ comparator to atomically compare these timestamps so that the operation with the lowest timestamp is eventually executed (see \Cref{sec:universal_construction} for more details).

%% file: figures/cas.tex
\begin{algorithm}[H]
\DontPrintSemicolon

\Init(\text{CAS($O$: object, $\old$: value, $\new$: value)}){

    $current \coloneqq O$

    \If{$current = \old$}{
        $O \coloneqq \new$

        \textbf{return} true
    }

    \textbf{return} false

}

\label{alg:cas}
\end{algorithm}

%% file: figures/gcas.tex
\begin{algorithm}[H]
\DontPrintSemicolon

\Init(\text{GCAS($comparator:$ value $\times$ value $\to$ boolean, $O$: object, $\old$: value, $\new$: value)}){

    $current \coloneqq O$

    \If{$comparator(current, \old)$}{
        $O \coloneqq \new$

        \textbf{return} true
    }

    \textbf{return} false

}

\label{alg:gcas}
\end{algorithm}

%% file: sections/model.tex
\section{Model Sketch}
\label{sec:model}

We consider an infinite arrival distributed system~\cite{merritt2000computing} where possibly infinitely many asynchronous processes that may fail by crashing communicate via shared objects such as \GCAS{} and \FA{}.
In such systems, shared objects can be used to \emph{implement} other shared objects such that the implemented objects are linearizable~\cite{herlihy1990linearizability} and wait-free~\cite{herlihy1991wait}.

In an implementation of an object $O$, each process $p$ executes \emph{steps} of three kinds: the invocation by $p$ of an operation $o$ on $O$ (we call $p$ the \emph{owner} of $o$); the receipt by $p$ of a response $r$ from $O$; and the \emph{low-level steps} required by the implementation to perform $o$ such as performing operations on the base \GCAS{} and \FA{} objects.
An \emph{implementation history} $I$ of $O$ is a sequence of invocation, response, and low-level steps such that for each process $p$, the subsequence of $I$ involving only the steps of $p$ follows the following alternating pattern.
An invocation step for some operation $o$, some number of low-level steps required by the implementation to perform $o$, and a response step for $o$.
For each process $p$, this pattern always begins with an invocation step but may end with either an invocation, response, or low-level step.
A \emph{history} $H$ of $O$ is the subsequence of an implementation $I$ of $O$ that only contains invocation and response steps.
Successive invocation and response steps in this subsequence are called \emph{matching.}
An \emph{operation execution} $opx$ in a history $H$ of an object is either a pair consisting of an invocation and its matching response in $H$, in which case we say that $opx$ is complete in $H$; or an invocation in $H$ that has no matching response in $H$, in which case we say that $opx$ is \emph{incomplete} or \emph{pending} in $H$.
A history $H$ is \emph{complete} if all operation executions in $H$ are complete.
A \emph{completion} of $H$ is a history $H'$ formed by removing invocations or adding responses to every pending operation execution in $H$ such that $H'$ is complete.

Each object has a type that specifics how the object behaves when it is accessed sequentially.
Formally, an \emph{object of type $\mathcal{T}$} is specified by a tuple $(OP, RES, Q, \delta)$, where $OP$ is a set of operations, $RES$ is a set of responses, $Q$ is a set of states, and $\delta \subseteq Q \times OP \times Q \times RES$ is a state transition relation.
A tuple $(s, o, s', r)$ in $\delta$ means that if type $\mathcal{T}$ is in state $s$ when operation $o \in OP$ is invoked, then $\mathcal{T}$ can change its state to $s'$ and return the response $r$.
Note that $\delta$ is a relation as opposed to a function to capture non-determinism.

When an implementation of an object $O$ of type $\mathcal{T}$ is accessed concurrently, its behavior should be \emph{linearizable with respect to $\mathcal{T}$:} every operation on $O$ must appear to take effect instantaneously, at some point during the operation's execution interval, according to type $\mathcal{T}$.
More precisely, a \emph{linearization $L$ of a complete history $H$ of an object} is an assignment of a distinct \emph{linearization point $L(opx)$} to every operation execution $opx$ in $H$ such that $L(opx)$ is within the execution interval of $opx$ in $H$.
A linearization $L$ of a complete history $H$ of an object \emph{conforms} to type $\mathcal{T}$, if the operation responses in $H$ could be those received when applying these operations sequentially, in the order dictated by $L$, on an object of type $\mathcal{T}$.
A history $H$ of an object is \emph{linearizable with respect to a type $\mathcal{T}$} if $H$ has a completion $H'$ and a linearization $L'$ of $H'$ that conforms to type $\mathcal{T}$.
An object $O$ is \emph{linearizable with respect to type $\mathcal{T}$} if every history $H$ of $O$ is linearizable with respect to type $\mathcal{T}$.

%% file: sections/universal_construction.tex
\section{A Simple Wait Free Universal Construction From GCAS}
\label{sec:universal_construction}

\Cref{alg:wait-free-simple} is a \emph{simple} implementation of an object $O$ of type $\mathcal{T}$ from \GCAS{} and \FA{} objects.
Before a process $p$ invokes its first operation, it obtains from the memory management system a single \GCAS{} object $\helpobject{p}$.
This object has two fields: one stores the timestamp of the operation on $O$ that $p$ wishes to execute, and the other eventually stores the response of that operation.
The \emph{response object} $\helpobject{p}$ is not known a prior to processes other than $p$.

In addition, the implementation uses three shared objects that are known a prior to any process that wishes to apply operations to $O$.
These three objects are:
\begin{enumerate}
    \item $\clockobject{}$: A \FA{} object used to timestamp operations; $\clockobject{}$ stands for \emph{clock}.
    \item $\announceobject{}$: A \GCAS{} object that stores information about the operation $o$ to be linearized next; $\announceobject{}$ stands for \emph{announce}. It has three fields: the timestamp of $o$, the operation $o$, and a pointer to the owner of $o$'s response object $\helpobject{}$.
    \item $\stateobject{}$: A \GCAS{} object that stores information about the state of the implemented object $O$; $\stateobject{}$ stands for \emph{state}. It has four fields: the timestamp of the last operation $o$ that resulted in the current state, the state of the object, the response of applying $o$, and a pointer to the owner of $o$'s response object $\helpobject{}$.
\end{enumerate}

To implement an object $O$ of type $\mathcal{T}$ we assume to be given the state transition relation $\delta$ of $\mathcal{T}$ in the form of a procedure $apply_{\mathcal{T}}(o, s)$.
This procedure takes an operation $o$ and current state $s$ as input and returns some $(s', r)$ as output, such that $(s, o, s', r) \in \delta$.
For convenience, we assume:

\begin{assumption}\label{assumption:1}
    $NULL$ is a value that differs from all possible responses to all operations of type $\mathcal{T}$ and $NOOP$ is a name that differs from the names of all operations of $\mathcal{T}$.
\end{assumption}

\input{figures/universal_construction}

We first explain the algorithm at a high level and then provide a line-by-line description.
To perform an operation $o$, a process $p$ does the following:
\begin{enumerate}
    \item $p$ gets a timestamp $t$ from the clock object $\clockobject{}$ for $o$.

    \item While $o$ is not done:
    \begin{enumerate}
        \item $p$ reads the current state $s$ of the object from the state object $\stateobject{}$.

        \item $p$ tries to put $o$ in the announce object $\announceobject{}$ using a GCAS($>, \ldots$) operation. This GCAS compares the timestamp $t$ of $o$ and the timestamp $t^*$ of the operation $o^*$ currently in the announce object. If $t < t^*$ then it replaces $o^*$ with $o$.

        \item Irrespective of whether this GCAS was successful (i.e., whether $p$ replaced $o^*$ with $o$), $p$ now helps to execute whatever operation is currently in the announce object. To do so:
        \begin{enumerate}
            \item $p$ reads the operation in the announce object; say it is $o'$.

            \item If $o'$ is not done, then $p$ first applies the operation $o'$ to the state $s$ of the object using the $apply_{\mathcal{T}}$ procedure (recall that $p$ read $s$ in step (a)) to determine what the new state $s'$ should be.
            It then uses a GCAS($=, \ldots$) to try and replace $s$ with $s'$ in the state object.

            \item If $o'$ is done, then $p$ tries to replace $o'$ with $o$ in the announce object by using a GCAS($=, \ldots$) operation. 
        \end{enumerate}
    \end{enumerate}
\end{enumerate}

It remains to explain how processes determine if an operation is done, and how a process learns the response of its own operation (which may have been previously executed by a ``helper'').
Process $p$ determines whether an operation $o$ with timestamp $t$ owned by a process $q$ is done as follows: $p$ reads the response object $\helpobject{q}$ of $q$; say it gets $(t',r')$.
$p$ knows that $o$ is done if either $t' = t$ and $r' \neq NULL$, or $t < t'$ (in which case $o$ is done because $q$ must have completed $o$ and then invoked an operation with a greater timestamp $t'$).
To get the response of its operation $p$ just reads the response field of its response object.
The response object of $p$ is maintained as follows.
When $p$ invokes an operation $o$ it sets its response object $\helpobject{p}$ to $(t,NULL)$ where $t$ is the timestamp of $o$.
Whenever a process tries to linearize an operation, it first copies the response of the \emph{last} linearized operation into the response object of the owner of that operation.

We now describe the algorithm line by line.
To perform an operation $o$ on $O$ process $p$ invokes DoOp($o$) on \cref{line:op_start}.
$p$ starts by getting a timestamp $t$ for $o$ by performing a \FA{} operation on $C$ on \cref{line:time_assignment}.
Then, $p$ sets its response object $\helpobject{p}$ to $(t, NULL)$ on \cref{line:reset_help_struct}.
From here onward, while $o$ is not done, $p$ tries to insert $o$ in $\announceobject{}$ and linearize the operation stored in $\announceobject{}$.
To do so, $p$ first reads the current contents $(t^*, s^*, r^*, roptr^*)$ from $\stateobject{}$ on \cref{line:g_r_query}.
$p$ then tries to copy $o^*$'s response (the operation that resulted in the state $s^*$) to the owner of $o^*$'s response object by performing a GCAS on $roptr^*$ on \cref{line:help_pointer_cas}.
Specifically it executes GCAS($=$, $roptr^*$, $(t^*, NULL)$, $(t^*, r^*)$) which will set the owner of $o^*$'s response object equal to $(t^*, r^*)$ only if it is still equal to $(t^*, NULL)$.
$p$ then tries to insert $o$ in $\announceobject{}$ on \cref{line:g_a_gcas} by comparing the current timestamp stored $\announceobject{}$ to the timestamp of $o$.
It does so by executing GCAS($>$, $\announceobject{}$, $(t, -, -)$, $(t, o, \&\helpobject{})$) which sets $\announceobject{}$ equal to $(t, o, \&\helpobject{})$ only if $o$'s timestamp is higher then the current timestamp of the operation stored in $\announceobject{}$.
$p$ then tries to linearize the operation stored in $\announceobject{}$.
It starts by reading the triplet $(t', o', roptr')$ from $\announceobject{}$ on \cref{line:g_a_query}.
To see whether $o'$ is not yet done, $p$ first reads the contents of $roptr'$ on \cref{line:help_pointer_query}, say $(\hat{t}, \hat{r})$, and then checks whether $(\hat{t}, \hat{r}) = (t', NULL)$ on \cref{line:help}.
If this condition is true, $o'$ is not done.
In this case, $p$ applies $o'$ to the current state $s^*$ of $O$ (which it read on \cref{line:g_r_query}) by executing $(s', r') \coloneqq apply_{\mathcal{T}}(o', s^*)$ on \cref{line:apply}.
$p$ then tries to change the current state of $O$ from $s^*$ to $s'$ by performing a GCAS on $\stateobject{}$ on \cref{line:g_r_cas}.
Specifically it executes GCAS($=$, $\stateobject{}$, $(t^*, s^*, r^*, roptr^*)$, $(t', s', r', roptr')$).
Note that $p$ read $(t^*, s^*, r^*, roptr^*)$ from $\stateobject{}$ on \cref{line:g_r_query}, $p$ read $t'$ and $roptr'$ from $\announceobject{}$ on \cref{line:g_a_query}, and $p$ got $s'$ and $r'$ by applying $o'$ to $s^*$ on \cref{line:apply}.
In the other case, i.e., when $(\hat{t}, \hat{r}) \neq (t', NULL)$ on \cref{line:help}, then $o'$ is done.
If so, $p$ tries to insert $o$ in $\announceobject{}$.
It does so by executing GCAS($=$, $\announceobject{}$, $(t', o', roptr')$, $(t, o, \&\helpobject{})$) on \cref{line:g_a_cas}.
Note that $p$ read $(t', o', roptr')$ from $\announceobject{}$ on \cref{line:g_a_query}.

%% file: figures/universal_construction.tex
\begin{algorithm}[t]
\DontPrintSemicolon

\nonl\textbf{Shared Object Per Participating Process $p$:}

\nonl\nonl\Init($\helpobject{p}:$ A \GCAS{} object with two fields:){
    \nonl$time:$ the timestamp of this process' last invoked operation.

    \nonl$response:$ the response of applying this operation or $NULL$.
}

\BlankLine

\nonl\textbf{Shared Objects:}

\nonl$\clockobject{}:$ A \FA{} object, initially 1.

\nonl\Init($\announceobject{}:$ A \GCAS{} object with three fields:){
    \nonl$time:$ the timestamp of $o$, initially 0.

    \nonl$o:$ the operation to execute, initially $NOOP$.

    \nonl$response\_object\_pointer:$ a pointer to the owner of $o$'s response object, initially $\helpobjectsymbol{}(NOOP)$ which is a pointer to an immutable location that contains $(0, \perp)$.
}

\nonl\Init($\stateobject{}:$ A \GCAS{} object with four fields:){
    \nonl$time:$ the timestamp of the operation $o$ that resulted in the current object state, initially 0.

    \nonl$state:$ the current object state, initially the starting state $s_0$ of the type $\mathcal{T}$.

    \nonl$response:$ the response of applying $o$ resulting in the current object state, initially $\perp$.

    \nonl$response\_object\_pointer:$ a pointer to the owner of $o$'s response object, initially $\helpobjectsymbol{}(NOOP)$.
}

\BlankLine

\nonl\textbf{Code for Process $p$:}

\BlankLine

\Init(\textbf{procedure} \text{DoOp($o$)}\label{line:op_start}){
    $t \coloneqq F\&I(\clockobject{})$\label{line:time_assignment}\tcc*[f]{get a timestamp $t$ for $o$}

    $\helpobject{p} \coloneqq (t, NULL)$\label{line:reset_help_struct}\tcc*[f]{initialize $o$'s response to $NULL$}

    \While(\tcc*[f]{while $o$ is not done}){$\helpobject{p} = (t, NULL)$}{\label{line:loop_start}
        $(t^*, s^*, r^*, roptr^*) \coloneqq \stateobject{}$\label{line:g_r_query}\tcc*[f]{read implemented object state}
        
        GCAS($=$, $roptr^*$, $(t^*, NULL)$, $(t^*, r^*)$)\label{line:help_pointer_cas}\tcc*[f]{copy response into response object}
    
        GCAS($>$, $\announceobject{}$, $(t, -, -)$, $(t, o, \&\helpobject{p})$)\label{line:g_a_gcas}\tcc*[f]{if $o$ has higher priority announce it}
    
        $(t', o', roptr') \coloneqq \announceobject{}$\label{line:g_a_query}\tcc*[f]{read operation $o'$ to be helped}
    
        $(\hat{t}, \hat{r}) \coloneqq (*roptr')$\label{line:help_pointer_query}\tcc*[f]{read response of $o'$}

        \If(\tcc*[f]{if $o'$ is not done}){$(\hat{t}, \hat{r}) = (t', NULL)$}{\label{line:help}
            $(s', r') \coloneqq apply_{\mathcal{T}}(o', s^*)$\label{line:apply}\tcc*[f]{apply $o'$ to object state}

            GCAS($=$, $\stateobject{}$, $(t^*, s^*, r^*, roptr^*)$, $(t', s', r', roptr')$)\label{line:g_r_cas}\tcc*[f]{try to linearize $o'$}
        }
        \Else(\tcc*[f]{if $o'$ is done}){
            GCAS($=$, $\announceobject{}$, $(t', o', roptr')$, $(t, o, \&\helpobject{p})$)\label{line:g_a_cas}\tcc*[f]{and $o'$ is still in $\announceobject{}$ announce $o$}
        }
    }

    \textbf{return} $\helpobject{p}.response$\label{line:op_done}
}

\caption{A Simple Implementation of an Object of Type $\mathcal{T}$ from GCAS and F\&I objects.}
\label{alg:wait-free-simple}
\end{algorithm}

%% file: sections/definitions_and_basic_facts.tex
\section{Definitions and Basic Facts}

Before we prove that \Cref{alg:wait-free-simple} is wait-free and linearizable we begin with some definitions and basic facts.
Throughout our proof, we assume an arbitrary implementation history $I$ of \Cref{alg:wait-free-simple}.
We use the notion of time to refer to how two steps occur within $I$.
Specifically, if steps $s_1$ and $s_2$ occur at $t_1$ and $t_2$ respectively such that $t_1 < t_2$ then $s_1$ comes before $s_2$ in $I$.
Finally, for brevity, we use the term operation to refer to an operation execution and execution of line $X$ to refer to a low-level step of line $X$ by some process in $I$.

\begin{definition}
    The invocation and response steps for an operation $o$ are \cref{line:op_start} and \cref{line:op_done}, respectively.
\end{definition}

\begin{definition}\label{def:op_state}
    For an operation $o \neq NOOP$: $p(o)$ denotes the process that executed $o$; $t(o)$ is equal to the response of the F\&I operation executed by $p(o)$ on \cref{line:time_assignment} in DoOp($o$) or $\infty$ if $p(o)$ has never executed \cref{line:time_assignment} in DoOp($o$); and $\helpobjectsymbol{}(o)$ denotes the response object $\helpobject{}$ of $p(o)$. For the special case of the $NOOP$ operation, $p(NOOP)$ is undefined, $t(NOOP) = 0$, and $\helpobjectsymbol{}(NOOP)$ points to an immutable location that contains $(0, \perp)$.
\end{definition}

\begin{definition}[Complete at $T$]\label{def:complete}
    An operation $o$ is \emph{complete at time $T$} if there exists a time $T' \leq T$ such that $p(o)$ executed \cref{line:op_done} at $T'$ within DoOp($o$).
\end{definition}

\begin{definition}[Done at $T$]\label{def:done}
    An operation $o$ is \emph{done at time $T$} if there exists a time $T' \leq T$ where $\helpobjectsymbol{}(o) = (t(o), r)$ at $T'$ such that $r \neq NULL$.
\end{definition}

We start with two simple observations about operations in general.
An immediate consequence of \Cref{def:done} is

\begin{observation}\label{observation:if_once_done_then_always_done}
    If operation $o$ is done at time $T$ then for all times $T' \geq T$ $o$ is done at $T'$.
\end{observation}

Since timestamp assignment is done using an atomic \FA{} (\cref{line:time_assignment}) that is initially 1, each operation is assigned a unique timestamp greater than or equal to 1. Hence

\begin{observation}\label{observation:unique_timestamps}
    For all operations $o$ and $o'$, if $t(o) \neq \infty$ and $t(o') \neq \infty$ then $o \neq o'$ if and only if $t(o) \neq t(o')$.
\end{observation}

The following two observations are about how $\announceobject{}$ and $\stateobject{}$ are well-formed.
Since $\announceobject{}$ is initially\\ $(t(NOOP), NOOP, \helpobjectsymbol{}(NOOP))$ and the new values passed on \cref{line:g_a_gcas} and \cref{line:g_a_cas} are always of the form $(t(o), o, \helpobjectsymbol{}(o))$, we have:

\begin{observation}\label{observation:g_a_is_well_formed}
    For all times $T$ there exists an operation $o$ such that $\announceobject{}$ equals $(t(o), o, \helpobjectsymbol{}(o))$ at $T$.
    Henceforth we abbreviate this as ``$(t(o), o, \helpobjectsymbol{}(o))$ is stored in $\announceobject{}$" or ``$o$ is stored in $\announceobject{}$" at time $T$.
    Furthermore, we say ``an execution of \cref{line:g_a_gcas} or \cref{line:g_a_cas} is for operation $o$" to mean that it is of the form $\text{GCAS(}-, \announceobject{}, -, (t(o), o, \helpobjectsymbol{}(o)))$.
\end{observation}

Likewise, $\stateobject{}$ is initially $(t(NOOP), s_0, \perp, \helpobjectsymbol{}(NOOP))$.
Since every new time and help pointer assigned to $\stateobject{}$ on \cref{line:g_r_cas} is supplied from reading $\announceobject{}$, it is always $t(o)$ and $\helpobjectsymbol{}(o)$ by \Cref{observation:g_a_is_well_formed}.
Moreover, the response field is set equal to the response from $apply_T$ which is always not $NULL$ by \Cref{assumption:1}.
Hence:

\begin{observation}\label{observation:g_r_is_well_formed}
    For all times $T$ there exists an operation $o$ such that $\stateobject{}$ equals $(t(o), -, r, \helpobjectsymbol{}(o))$ at $T$ for some response $r \neq NULL$.
    Henceforth we abbreviate this as ``$(t(o), -, r, \helpobjectsymbol{}(o))$ is stored in $\stateobject{}$" or ``$o$ is stored in $\stateobject{}$" at time $T$.
    Furthermore, we say ``an execution of \cref{line:g_r_cas} is for some operation $o$" to mean that it is of the form $\text{GCAS}(=, \stateobject{}, -, (t(o), -, r, \helpobjectsymbol{}(o)))$.
\end{observation}

The following two observations concern each process $p$'s response object $\helpobject{p}$.

\begin{observation}\label{observation:help_pointer_cas_is_well_formed}
    For each execution of \cref{line:help_pointer_cas} there exists an operation $o$ where said execution is of the form $\text{GCAS}(=, \helpobjectsymbol{}(o), (t(o), NULL), (t(o), r))$ for some response $r \neq NULL$.
    Henceforth we abbreviate this as ``\cref{line:help_pointer_cas} is executed for $o$".
\end{observation}

\begin{observation}\label{observation:only_owner_changes_times_of_help_struct}
    For all processes $p$ the following hold:
    \begin{itemize}
        \item $\helpobject{p}.time$ only changes by $p$ executing \cref{line:reset_help_struct}.
        \item Every execution of \cref{line:reset_help_struct} by $p$ sets $\helpobject{p}.time$ to a unique value.
        \item $\helpobject{p}.time$ is monotonically increasing.
    \end{itemize}
\end{observation}

Using these definitions and observations we now prove some basic facts about our universal construction.

\begin{lemma}\label{lemma:once_response_is_not_null_only_owner_sets_to_null}
    Suppose a process $p$'s $\helpobject{p} = (t, r)$ at time $T$ where $r \neq NULL$.
    For all times $T' \geq T$ if $\helpobject{p}.time = t$ at $T'$ then $\helpobject{p}.response = r$ at $T'$.
\end{lemma}

\begin{proof}
    Suppose that $\helpobject{p} = (t, r)$ at time $T$ where $r \neq NULL$, and $\helpobject{p}.time = t$ at time $T' \geq T$.
    Since by \Cref{observation:only_owner_changes_times_of_help_struct}, $\helpobject{p}.time$ is monotonically increasing, this implies $\helpobject{p}.time$ equals $t$ throughout $[T, T']$.
    Since by \Cref{observation:only_owner_changes_times_of_help_struct}, $\helpobject{p}.time$ only changes by $p$ executing \cref{line:reset_help_struct}, and every execution of \cref{line:reset_help_struct} by $p$ sets $\helpobject{p}.time$ to a unique value, this implies \cref{line:reset_help_struct} has not been executed by $p$ throughout $[T, T']$.
    Thus throughout $[T, T']$ the contents of $\helpobject{p}$ can only changed by \cref{line:help_pointer_cas}.
    However, since \cref{line:help_pointer_cas} only changes $\helpobject{p}.response$ if it equals $NULL$, and $\helpobject{p}.response = r$ which is not $NULL$ at $T$, all executions of \cref{line:help_pointer_cas} throughout $[T, T']$ return false.
    Therefore, $\helpobject{p}.response = r$ throughout $[T, T']$.
    \qH{\Cref{lemma:once_response_is_not_null_only_owner_sets_to_null}}
\end{proof}

\begin{lemma}\label{lemma:complete_implies_done}
    If operation $o$ is complete at time $T$ then $o$ is done at some time $T' < T$.
\end{lemma}

\begin{proof}
    If operation $o$ is complete at time $T$ then by \Cref{def:complete} there exists a time $T_1 \leq T$ where $p(o)$ executed \cref{line:op_done} at time $T_1$.
    Hence, \cref{line:loop_start} was false sometime $T_2 < T_1$.
    Since by \Cref{observation:only_owner_changes_times_of_help_struct} \cref{line:reset_help_struct} is the only step that changes the time field of $\helpobjectsymbol{}(o)$, this implies the time field equals $t(o)$ at $T_2$.
    Thus, at $T_2$, $\helpobjectsymbol{}(o) = (t(o), r)$ where $r \neq NULL$.
    Therefore, by \Cref{def:done} $o$ was done at $T_2$ which is before $T$ by transitivity.
    \qH{\Cref{lemma:complete_implies_done}}
\end{proof}

\begin{lemma}\label{lemma:help_struct_change_implies_done}
    For all operations $o$, if $\helpobjectsymbol{}(o) = (t(o), -)$ at time $T$ and $\helpobjectsymbol{}(o) = (t, -)$ such that $t \neq t(o)$ at time $T' > T$ then there exists a time $T^* < T'$ where $o$ is done at $T^*$.
\end{lemma}

\begin{proof}
    Since by \Cref{observation:only_owner_changes_times_of_help_struct} the time field of $\helpobjectsymbol{}(o)$ is monotonically increasing, $t \neq t(o)$, and $T' > T$, $t(o) < t$.
    Hence, $p(o)$ executed \cref{line:reset_help_struct} for some operation $o' \neq o$ where $t = t(o')$ at time $T_1 \leq T'$.
    Thus $p(o)$ completed $o$ at some time $T_2 < T_1$.
    Therefore \Cref{lemma:complete_implies_done} implies $o$ is done at $T_2$ which is before $T'$ by transitivity.
    \qH{\Cref{lemma:help_struct_change_implies_done}}
\end{proof}

\begin{lemma}\label{lemma:copy_help_struct_implies_done}
    If a process executes \cref{line:help_pointer_cas} for operation $o$ at time $T$ then $o$ is done at $T$.
\end{lemma}

\begin{proof}
    Suppose a process $p$ executes \cref{line:help_pointer_cas} for $o$ at some time $T$.
    If $o = NOOP$ then by \Cref{def:op_state} $\helpobjectsymbol{}(NOOP) = (0, \perp)$ always.
    Hence by \Cref{def:done}, $o$ is done.
    
    Now suppose $o \neq NOOP$.
    By \Cref{observation:help_pointer_cas_is_well_formed} $p$'s execution of \cref{line:help_pointer_cas} at time $T$ was of the form 
    
    \noindent $\text{GCAS}(=, \helpobjectsymbol{}(o), (t(o), NULL), (t(o), r))$ for some $r \neq NULL$.
    Consider the response of this GCAS.
    If it is true this implies $\helpobjectsymbol{}(o) = (t(o), r)$ at time $T$.
    Hence by \Cref{def:done} $o$ is done at $T$.
    Otherwise, the response of $p$'s GCAS is false.
    This implies $\helpobjectsymbol{}(o) = (t', r')$ such that $t' \neq t(o)$ or $r' \neq NULL$ at time $T$.
    Suppose $t' = t(o)$.
    Hence, $r' \neq NULL$.
    Thus, by \Cref{def:done} $o$ is done at $T$.
    Now suppose $t' \neq t(o)$.
    Since $p$ executed \cref{line:help_pointer_cas} for $o$, this implies $\helpobjectsymbol{}(o)$, $t(o)$, and $r$ were stored in $\stateobject{}$ at time $T_1 < T$ when $p$ executed \cref{line:g_r_query}.
    Since $o \neq NOOP$ this implies some process $p'$ executed \cref{line:g_r_cas} for $o$ that returned true at time $T_2 < T_1$.
    Hence, $p'$ executed \cref{line:help} whose condition was true at time $T_3 < T_2$.
    Thus, when $p'$ execute \cref{line:help_pointer_query} at time $T_4 < T_3$ $\helpobjectsymbol{}(o) = (t(o), NULL)$.
    Since $\helpobjectsymbol{}(o) = (t(o), NULL)$ at $T_4 < T$ and $\helpobjectsymbol{}(o) = (t', r')$ at $T$ where $t' \neq t(o)$, by \Cref{lemma:help_struct_change_implies_done} $o$ is done at some time $T_5 < T$.
    Therefore by \Cref{observation:if_once_done_then_always_done} $o$ is done at $T$.
    \qH{\Cref{lemma:copy_help_struct_implies_done}}
\end{proof}

\begin{lemma}\label{lemma:non_null_help_struct_response}
    Suppose $\stateobject{}$ stores operations $o$ and $o'$ at times $T$ and $T' > T$, respectively.
    If $o \neq o'$ then $o$ is done at $T'$.
\end{lemma}

\begin{proof}
    Since $\stateobject{}$ stores $o$ and later it stores $o' \neq o$, $o' \neq NOOP$. So, some process $p$ executed \cref{line:g_r_cas} for $o'$ at some time $T_1 \leq T'$ and the GCAS on \cref{line:g_r_cas} returned true.
    Since $T' > T$, $o$ is stored in $\stateobject{}$ at $T$, and $o'$ is stored in $\stateobject{}$ at $T'$, $T_1 \geq T$.
    Without loss of generality suppose $T_1$ was the first time from $T$ onwards where \cref{line:g_r_cas} returned true.
    Hence, $p$ read ($t(o)$, $-$, $r$, $\helpobjectsymbol{}(o)$) in $\stateobject{}$ on \cref{line:g_r_query} at some time $T_2 \in (T, T_1)$.
    Thus $p$ executed \cref{line:help_pointer_cas} for $o$ at some time $T_3 \in (T_2, T_1)$.
    By \Cref{lemma:copy_help_struct_implies_done} this implies $o$ is done at $T_3$.
    By transitivity $T_3 < T'$.
    Therefore, by \Cref{observation:if_once_done_then_always_done} $o$ is done at $T'$.
    \qH{\Cref{lemma:non_null_help_struct_response}}
\end{proof}

%% file: sections/proof_of_wait_freedom.tex
\section{Proof of Wait-Freedom}

In this section, we prove that \Cref{alg:wait-free-simple} is wait-free, that is:
\begin{definition}[Wait-freedom]
    An implementation of an object is wait-free if a process cannot invoke an operation and then take infinitely many steps without completing it.
\end{definition}

We start by proving some intermediate lemmas.

\begin{lemma}\label{lemma:not_done_implies_help_struct_fixed}
    For all operations $o$, if $\helpobjectsymbol{}(o) = (t(o), NULL)$ at time $T$ and $o$ is not done at $T' \geq T$ then $\helpobjectsymbol{}(o) = (t(o), NULL)$ throughout $[T, T']$.
\end{lemma}

\begin{proof}
    Suppose that $\helpobjectsymbol{}(o) = (t(o), NULL)$ at $T$ and $o$ is not done at time $T' \geq T$.
    By \Cref{lemma:complete_implies_done} $o$ is not complete at $T'$.
    Hence, $p(o)$ has not executed \cref{line:reset_help_struct} throughout $(T, T']$.
    Since by \Cref{observation:only_owner_changes_times_of_help_struct}, \cref{line:reset_help_struct} is the only step that changes the $time$ field of $\helpobjectsymbol{}(o)$, the $time$ field of $\helpobjectsymbol{}(o)$ equals $t(o)$ throughout $[T, T']$.
    Finally, since $o$ is not done at $T'$, and the $time$ field of $\helpobjectsymbol{}(o)$ equals $t(o)$ throughout $[T, T']$, by \Cref{def:done} the response field of $\helpobjectsymbol{}(o)$ equals $NULL$ throughout $[T, T']$.
    Therefore, $\helpobjectsymbol{}(o) = (t(o), NULL)$ throughout $[T, T']$.
    \qH{\Cref{lemma:not_done_implies_help_struct_fixed}}
\end{proof}

\begin{lemma}\label{lemma:if_observed_done_wont_try_to_do}
    Suppose a process $p$ reads $(t(o), o, \helpobjectsymbol{}(o))$ in $\announceobject{}$ on \cref{line:g_a_query} at time $T$ and $o$ is done at that time.
    If $p$ executes \cref{line:help} after $T$ then the condition of \cref{line:help} is false.
\end{lemma}

\begin{proof}
    Since $o$ is done at $T$ there exists a time $T_1 \leq T$ where $\helpobjectsymbol{}(o) = (t(o), r)$ such that $r \neq NULL$.
    Furthermore, since $p$ reads $(t(o), o, \helpobjectsymbol{}(o))$ in $\announceobject{}$ on \cref{line:g_a_query} at $T$, $p$ reads $(\hat{t}, \hat{r})$ in $\helpobjectsymbol{}(o)$ on \cref{line:help_pointer_query} at $T_2 > T$.
    Since by \Cref{observation:only_owner_changes_times_of_help_struct} the time field of $\helpobjectsymbol{}(o)$ is monotonically increasing, $\hat{t} \geq t(o)$.
    In \cref{line:help}, $p$ checks whether $(\hat{t}, \hat{r}) = (t(o), NULL)$.
    If $\hat{t} = t(o)$ then by \Cref{lemma:once_response_is_not_null_only_owner_sets_to_null} $\hat{r} = r \neq NULL$.
    Hence, the condition of \cref{line:help} is false.
    Otherwise, $\hat{t} > t(o)$, and so the condition of \cref{line:help} is also false.
    \qH{\Cref{lemma:if_observed_done_wont_try_to_do}}
\end{proof}

\newcommand{\T}{\bar{T}}

\begin{theorem}\label{theorem:wait_free}
    \Cref{alg:wait-free-simple} is wait-free.
\end{theorem}

\begin{proof}

Suppose for contradiction that \Cref{alg:wait-free-simple} is not wait-free.
This implies that there is an operation $o$ where $p(o)$ takes infinitely many steps in DoOp($o$) but $o$ does not complete. Let:
\begin{definition}\label{def:stuck}
    $O = \{o\ \vert\ \text{$p(o)$ takes infinitely many steps in DoOp($o$) but $o$ does not complete}\}$
\end{definition}

\begin{definition}\label{def:omin}
    $o_{\min}$ be the operation with the minimum timestamp in the set $O$, i.e., $\forall o \in O, t(o_{\min}) \leq t(o).$
\end{definition}

We start with an observation about operations in $O$.
If $o \in O$ then $p(o)$ is stuck in the loop in \cref{line:loop_start}-\ref{line:g_a_cas}.
Thus, infinitely often the response field of $\helpobjectsymbol{}(o)$ is $NULL$.
Furthermore, since by \Cref{observation:only_owner_changes_times_of_help_struct} the time field of $\helpobjectsymbol{}(o)$ only changes on \cref{line:reset_help_struct} and $p(o)$ is stuck in the loop, the time field of $\helpobjectsymbol{}(o)$ is $t(o)$ from the time $p(o)$ executes \cref{line:reset_help_struct} during the execution of DoOp($o$) onwards.
Thus, by \Cref{lemma:once_response_is_not_null_only_owner_sets_to_null} $\helpobjectsymbol{}(o) = (t(o), NULL)$ from this time onwards. Therefore,

\begin{observation}\label{observation:stuck_implies_not_done}
    No operation in $O$ is ever done.
\end{observation}

We prove \Cref{theorem:wait_free} by showing that $o_{\min}$ is done at some time, contradicting \Cref{observation:stuck_implies_not_done}.
We do so in two parts.
The first shows that if an operation $o$ is stored in $\announceobject{}$ from some time onwards then $o$ is done at some time.
Roughly speaking this is because from some time onwards all executions of \cref{line:g_r_cas} are for $o$ implying one of them succeeds.
The second part shows that $o_{\min}$ is stored in $\announceobject{}$ from some time onwards.
Roughly speaking this is because eventually all operations with a smaller assigned timestamp than $o_{\min}$ complete or crash and ops with larger timestamps cannot ``dislodge" $o_{\min}$ from $\announceobject{}$ because of the use of GCAS($>$, $\ldots$) on \cref{line:g_a_gcas}.
Together these imply $o_{\min}$ is done.

\begin{claim}{\ref{theorem:wait_free}.1}\label{claim:stuck_implies_done}
    If operation $o$ is stored in $\announceobject{}$ at all times from some time $T$ onwards then $o$ is eventually done.
\end{claim}

\begin{proof}
    Let $o$ be an operation that is stored in $\announceobject{}$ from time $T$ onwards, and suppose for contradiction that $o$ is never done. This implies the following.

    \begin{claimindented}{\ref{claim:stuck_implies_done}.1}\label{claim:o_help_struct_fixed}
        From time $T$ onwards $\helpobjectsymbol{}(o) = (t(o), NULL)$.
    \end{claimindented}

    \begin{proofindented}
        By the assumption of \Cref{claim:stuck_implies_done} $o$ is written into $\announceobject{}$ at or before $T$.
        Hence, $p(o)$ wrote $o$ into $\announceobject{}$ in \cref{line:g_a_gcas} or \cref{line:g_a_cas} at or before $T$.
        Thus, $p(o)$ executed \cref{line:reset_help_struct} at some time $T_1 < T$ implying $\helpobjectsymbol{}(o) = (t(o), NULL)$ at $T_1$.
        Since $o$ is never done, \Cref{lemma:not_done_implies_help_struct_fixed} implies $\helpobjectsymbol{}(o) = (t(o), NULL)$ from $T_1$ onwards.
        Therefore, since $T_1 < T$, $\helpobjectsymbol{}(o) = (t(o), NULL)$ from $T$ onwards as wanted.
        \qH{\Cref{claim:o_help_struct_fixed}}
    \end{proofindented}

    \begin{claimindented}{\ref{claim:stuck_implies_done}.2}\label{claim:always_trytodo}
        If a process $p$ executes \cref{line:g_a_query} at or after $T$ followed by executing \cref{line:g_r_cas} then said execution of \cref{line:g_r_cas} is for $o$.
    \end{claimindented}

    \begin{proofindented}
        Suppose $p$ executed \cref{line:g_a_query} at time $T_1 \geq T$ followed by executing \cref{line:g_r_cas}.
        Hence, $p$ executed \cref{line:help_pointer_query}, \ref{line:help}, and \ref{line:g_r_cas} at times $T_2$, $T_3$, and $T_4$, respectively, such that $T \leq T_1 < T_2 < T_3 < T_4$.
        Since $p$ executed \cref{line:g_a_query} at $T_1 \geq T$, $o$ is stored in $\announceobject{}$ at $T_1$.
        Hence, by \Cref{observation:g_a_is_well_formed} $p$ reads $(t(o), o, \helpobjectsymbol{}(o))$ in $\announceobject{}$ on \cref{line:g_a_query} at $T_1$.
        By \Cref{claim:o_help_struct_fixed} $p$ reads ($t(o)$, $NULL$) in $h(o)$ on \cref{line:help_pointer_query} at $T_2$.
        Consequently, \cref{line:help} is true at $T_3$.
        Therefore, $p$ executes \cref{line:g_r_cas} for $o$ at $T_4$ as wanted.
        \qH{\Cref{claim:always_trytodo}}
    \end{proofindented}
    
    \noindent We now show that this implies there is a time after which all executions of \cref{line:g_r_cas} are for $o$.
    
    \begin{claimindented}{\ref{claim:stuck_implies_done}.3}\label{claim:eventually_no_other_ops}
        There exists a time $T_1 \geq T$ such that from $T_1$ onwards all executions of \cref{line:g_r_cas} are for $o$.
    \end{claimindented}

    \begin{proofindented}
        Suppose for contradiction that for all times $T_1 \geq T$ there exists a time $T_2 \geq T_1$ where \cref{line:g_r_cas} is executed at $T_2$ for an operation other than $o$.
        Thus, there are infinitely many executions of \cref{line:g_r_cas} for an operation other than $o$. 
        This implies that infinitely many reads of $\announceobject{}$ in \cref{line:g_a_query} find that $\announceobject{}$ stores an operation other than $o$.
        This contradicts the assumption of \Cref{claim:stuck_implies_done}.
        \qH{\Cref{claim:eventually_no_other_ops}}
    \end{proofindented}

    \noindent We now prove the final intermediate claim which states that winning the GCAS on \cref{line:g_r_cas} for an operation $o$ implies $o$ is done.

    \begin{claimindented}{\ref{claim:stuck_implies_done}.4}\label{claim:g_r_cas_true_implies_done}
        If the execution of GCAS on $\stateobject{}$ on \cref{line:g_r_cas} at some time $T_2$ is for $o$ and returns true then $o$ is done at some time $T_3 > T_2$.
    \end{claimindented}
    
    \begin{proofindented}
        Suppose \cref{line:g_r_cas} is executed for $o$ and it returns true at time $T_2$.
        Consider the first iteration of the loop by $p(o_{\min})$ after $T_2$ which occurs by \Cref{def:stuck}.
        There are two cases:
        \begin{enumerate}
            \item[] \hspace{-15pt}\textbf{Case 1.} $p(o_{\min})$ reads ($t(o)$, $-$, $r$, $\helpobjectsymbol{}(o)$) in $\stateobject{}$ on \cref{line:g_r_query}.
            
            Hence, $p(o_{\min})$ executes \cref{line:help_pointer_cas} for $o$ at $T_3 > T_2$.
            By \Cref{lemma:copy_help_struct_implies_done} $o$ is done at $T_3$.
            
            \item[] \hspace{-15pt}\textbf{Case 2.} $p(o_{\min})$ reads a different value in $\stateobject{}$ on \cref{line:g_r_query}.
            
            Hence, by \Cref{observation:g_r_is_well_formed} $\stateobject{}$ stores an operation other than $o$ at some time $T_3 > T_2$.
            By \Cref{lemma:non_null_help_struct_response} $o$ is done at $T_3$.
        \end{enumerate}
        Therefore in all cases, $o$ is done at some time $T_3 > T_2$.
        \qH{\Cref{claim:g_r_cas_true_implies_done}}
    \end{proofindented}

    \input{figures/claim_1_proof}

    We now complete the proof of \Cref{claim:stuck_implies_done}.
    Consider the first time $p(o_{\min})$ executes \cref{line:g_r_query} after $T_1$.
    Since $p(o_{\min})$ takes infinitely many steps in DoOp($o_{\min}$), it will execute \cref{line:g_a_query} afterward followed by existing \cref{line:g_r_cas}.
    Since $T_1 \geq T$, by \Cref{claim:always_trytodo} $p(o_{\min})$ executes \cref{line:g_r_cas} for $o$.
    If the GCAS on $\stateobject{}$ on \cref{line:g_r_cas} returns true then by \Cref{claim:g_r_cas_true_implies_done} $o$ is done sometime after.
    Otherwise, \cref{line:g_r_cas} returns false.
    Since \cref{line:g_r_cas} returns false some other process $p$ executed a GCAS on $\stateobject{}$ on \cref{line:g_r_cas} which returned true in between the time when $p(o_{\min})$ executed lines \ref{line:g_r_query} and \ref{line:g_r_cas}.
    Since $p(o_{\min})$ executed \cref{line:g_r_query} after $T_1$, $p$ executed \cref{line:g_r_cas} at $T_2 > T_1$.
    Thus, by \Cref{claim:eventually_no_other_ops} $p$ executed the GCAS on $\stateobject{}$ on \cref{line:g_r_cas} for $o$.
    Since that GCAS returned true, by \Cref{claim:g_r_cas_true_implies_done} $o$ is done sometime after $T_2$.
    Therefore in all cases, $o$ is done at some time, a contradiction.
    \qH{\Cref{claim:stuck_implies_done}}
\end{proof}

We now prove the second part, that $o_{\min}$ is stored in $\announceobject{}$ from some time onwards.
We start with an observation regarding operations with timestamps less than $o_{\min}$. Let:
\begin{definition}\label{def:less_than_omin}
    $O_< = \{o\ \vert\ t(o) < t(o_{\min})\}$.
\end{definition}
By \Cref{def:stuck} $o \in O$ if and only if $p(o)$ takes infinitely many steps in DoOp($o$) but $o$ does not complete.
By \Cref{def:omin}, $o_{\min}$ has the smallest timestamp in $O$.
So any process $p$ that invokes any operation $o'$ with a strictly smaller timestamp only takes finitely many steps in DoOp($o'$), because either $p$ crashes while executing DoOp($o'$) or because $p$ completes $o'$. Hence,

\begin{observation}\label{observation:crash_and_done}
    For all $o \in O_<$ there is a time after which $p(o)$ does not take a step during DoOp($o$).
\end{observation}

\begin{claim}{\ref{theorem:wait_free}.2}\label{claim:o_min_always}
    There exists a time $T$ such that for all times $T' \geq T$, $o_{\min}$ is stored in $\announceobject{}$ at $T'$.
\end{claim}

\begin{proof}
    We start by showing the following:
    
    \begin{claimindented}{\ref{claim:o_min_always}.1}\label{claim:all_in_o_get_out}
        There is a time $\T_1$ after which no operation $o \in O_<$ is stored in $\announceobject{}$.
    \end{claimindented}
    
    \begin{proofindented}
        Suppose for contradiction that for all times there is a later time when some $o \in O_<$ is stored in $\announceobject{}$.
        There are two cases.
        \begin{enumerate}

            \item[] \hspace{-15pt}\textbf{Case 1.} The operation stored in $\announceobject{}$ changes infinitely often.
            
            Since $O_<$ is finite, some operation $o \in O_<$ is written into $\announceobject{}$ infinitely often.
            Hence, $p(o)$ executes GCAS on $\announceobject{}$ on \cref{line:g_a_gcas} or \cref{line:g_a_cas} infinitely often during the execution of DoOp($o$), contradicting \Cref{observation:crash_and_done}.
        
            \item[] \hspace{-15pt}\textbf{Case 2.} The operation in $\announceobject{}$ changes finitely often.
            
            Since infinitely often some operation in $O_<$ is stored in $\announceobject{}$, this implies there exists an operation $o \in O_<$ which is stored in $\announceobject{}$ from some time $T$ onwards.
            By \Cref{claim:stuck_implies_done}, $o$ is done at some time $D$.
            Consider the first time $p(o_{\min})$ reads $\announceobject{}$ on \cref{line:g_a_query} after $\max(T, D)$; this must occur by \Cref{def:stuck}.
            Since $p(o_{\min})$'s execution of \cref{line:g_a_query} is after $T$, $p(o_{\min})$ read $(t(o), o, \helpobjectsymbol{}(o))$ in $\announceobject{}$.
            Moreover, since this execution is after $D$, $o$ is done by \Cref{observation:if_once_done_then_always_done}.
            Hence, by \Cref{lemma:if_observed_done_wont_try_to_do} $p(o_{\min})$ finds the condition of \cref{line:help} to be false.
            Thus $p(o_{\min})$ executes the GCAS on $\stateobject{}$ on \cref{line:g_a_cas}.
            If this returns true then sometime after $T$ $o_{\min}$ is stored in $\announceobject{}$ contradicting that $o$ is stored in $\announceobject{}$ from $T$ onwards.
            Otherwise \cref{line:g_a_cas} returns false.
            This implies that between when $p(o_{\min})$ executed lines \ref{line:g_a_query} and \ref{line:g_a_cas} $\announceobject{}$ stored an operation other than $o$.
            Since $p(o_{\min})$ executed \cref{line:g_a_query} at time $\max(T, D)$, $\announceobject{}$ stored an operation other than $o$ after $T$.
            This contradicts that $\announceobject{}$ stores $o$ from $T$ onwards.
        \end{enumerate}
        \qH{\Cref{claim:all_in_o_get_out}}
    \end{proofindented}

    \input{figures/claim_2_1_case_2_proof}
    
    \begin{claimindented}{\ref{claim:o_min_always}.2}\label{claim:o_min_writes}
        There exists a time $\T_2 \geq \T_1$ where $o_{\min}$ is stored in $\announceobject{}$ at $\T_2$.
    \end{claimindented}
    
    \begin{proofindented}
        Consider the first time $p(o_{\min})$ executes the GCAS on $\announceobject{}$ on \cref{line:g_a_gcas} at $\T_2 \geq \T_1$.
        If this GCAS returns true then $o_{\min}$ is stored in $\announceobject{}$ at $\T_2$.
        Otherwise, by the definition of GCAS, at time $\T_2$ some operation $o$ is stored in $\announceobject{}$ where $t(o) \leq t(o_{\min})$.
        However, since this execution is at or after $\T_1$, \Cref{claim:all_in_o_get_out} implies $t(o) \geq t(o_{\min})$.
        Thus $t(o) = t(o_{\min})$.
        Hence, by \Cref{observation:unique_timestamps} $o = o_{\min}$.
        Therefore, $o_{\min}$ is stored in $\announceobject{}$ at~$\T_2$.
        \qH{\Cref{claim:o_min_writes}}
    \end{proofindented}

    \input{figures/claim_2_proof}

    We now complete the proof of \Cref{claim:o_min_always}.
    Suppose for contradiction that for all times there exists a later time when $o_{\min}$ is not stored in $\announceobject{}$.
    Hence, there exists a time $T > \T_2$ where $o_{\min}$ is not stored in $\announceobject{}$ at $T$.
    Without loss of generality suppose this is the \emph{first time} after $\T_2$ where $o_{\min}$ is not stored in $\announceobject{}$.
    For $o_{\min}$ to not be stored in $\announceobject{}$ at $T$, a process that invoked an operation $o \neq o_{\min}$ must have executed a GCAS on $\announceobject{}$ on \cref{line:g_a_gcas} or \cref{line:g_a_cas} for $o$ which returned true at $T$.
    Since $T > \T_2 \geq \T_1$ by \Cref{claim:all_in_o_get_out} $o \notin O_<$.
    Hence, by \Cref{def:less_than_omin} $t(o) \geq t(o_{\min})$, and since $o \neq o_{\min}$, $t(o) > t(o_{\min})$ by \Cref{observation:unique_timestamps}.
    Consequently, the GCAS on $\announceobject{}$ on \cref{line:g_a_gcas} could not have returned true for $p(o)$ at $T$.
    Hence, it was a GCAS on $\announceobject{}$ on \cref{line:g_a_cas} that returned true for $p(o)$ at $T$.
    Since time $T$ is the first to change $\announceobject{}$ after $\T_2$, $p(o)$ read $(t(o_{\min}), o_{\min}, \helpobjectsymbol{}(o_{\min}))$ in $\announceobject{}$ on \cref{line:g_a_query}.
    Hence, $p(o)$ read $\helpobjectsymbol{}(o_{\min}) = (\hat{t}, \hat{r})$ on \cref{line:help_pointer_query}.
    Since $p(o)$ executed \cref{line:g_a_cas}, the condition of \cref{line:help} must have been false.
    Thus, $(\hat{t}, \hat{r}) \neq (t(o_{\min}), NULL)$.
    There are two cases.
    \begin{enumerate}
        \item[] \hspace{-15pt}\textbf{Case 1.} $\hat{t} = t(o_{\min})$.
        
        Hence, $\hat{r} \neq NULL$.
        Thus, $o_{\min}$ is done at some time.

        \item[] \hspace{-15pt}\textbf{Case 2.} $\hat{t} \neq t(o_{\min})$.
        
        Since $o_{\min}$ is stored in $\announceobject{}$ at $\T_2$, $p(o_{\min})$ executed the GCAS on $\announceobject{}$ on \cref{line:g_a_gcas} or \cref{line:g_a_cas} before or at $\T_2$.
        Hence, $p(o_{\min})$ set $h(o_{\min})$ to $(t(o_{\min}), NULL)$ on \cref{line:reset_help_struct} before $\T_2$.
        Thus, by \Cref{lemma:help_struct_change_implies_done} $o_{\min}$ is done sometime before $\T_2$.
    \end{enumerate}
    Therefore in all cases, $o_{\min}$ is done at some time, a contradiction to \Cref{observation:stuck_implies_not_done}.
    \qH{\Cref{claim:o_min_always}}
\end{proof}

Now return to the proof of \Cref{theorem:wait_free}.
Since \Cref{claim:o_min_always} satisfies the assumption of \Cref{claim:stuck_implies_done} $o_{\min}$ is done at some time, a contradiction to \Cref{observation:stuck_implies_not_done}.
\qH{\Cref{theorem:wait_free}}
\end{proof}

%% file: figures/claim_1_proof.tex
\def\length{13}
\def\start{3.5}

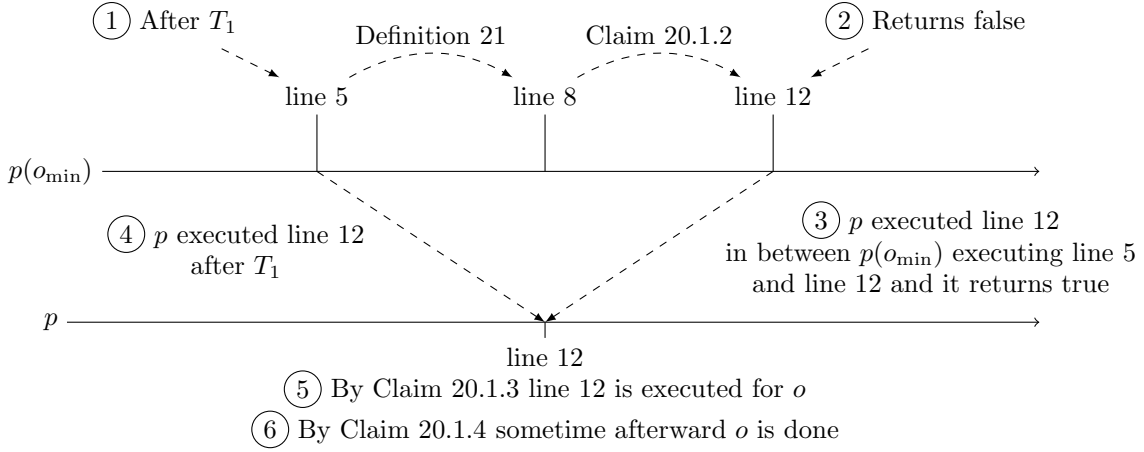
\begin{figure}
    \centering
    \begin{tikzpicture}
        \node (a) at (0, 2) {$p(o_{\min})$};
    
        \draw[->] (a.east) -- (\length, 2);

        \node (a) at (0, 0) {$p$};

        \draw[->] (a.east) -- (\length, 0);

        \node (a) at (\start, 3) {\cref{line:g_r_query}};

        \draw[-] (\start, 2) -- (a.south);

        \node (b) at (\start - 2, 3 + 1) {\circled{1} After $T_1$};

        \path[-latex, dashed] (b) edge (a);

        \node (b) at (1/2*\length, 3) {\cref{line:g_a_query}};

        \draw[-] (1/2*\length, 2) -- (b.south);

        \path[-latex, dashed] (a) edge [bend left] node[above] {\Cref{def:stuck}} (b);

        \node (c) at (\length - \start, 3) {\cref{line:g_r_cas}};

        \draw[-] (\length - \start, 2) -- (c.south);

        \path[-latex, dashed] (b) edge [bend left] node[above] {\Cref{claim:always_trytodo}} (c);

        \node (b) at (\length - \start + 2, 3 + 1) {\circled{2} Returns false};

        \path[-latex, dashed] (b) edge (c);

        \node[align=center] (a) at (1/2*\length, -1) {\cref{line:g_r_cas}\\ \circled{5} By \Cref{claim:eventually_no_other_ops} \cref{line:g_r_cas} is executed for $o$\\ \circled{6} By \Cref{claim:g_r_cas_true_implies_done} sometime afterward $o$ is done};

        \draw[-] (1/2*\length, 0) -- (a.north);
        \draw[-latex, dashed] (\length - \start, 2) -- node[right=.75cm, align=center] {\circled{3} $p$ executed \cref{line:g_r_cas}\\ in between $p(o_{\min})$ executing \cref{line:g_r_query}\\ and \cref{line:g_r_cas} and it returns true} (1/2*\length, 0);

        \draw[-latex, dashed] (\start, 2) -- node[left=.75cm, align=center] {\circled{4} $p$ executed \cref{line:g_r_cas}\\ after $T_1$} (1/2*\length, 0);
    \end{tikzpicture}
    \caption{Run constructed in the proof of \Cref{claim:stuck_implies_done} for the case where $p(o_{\min})$ receives false from the GCAS on \cref{line:g_r_cas}. \protect\circled{2} implies \protect\circled{3}. \protect\circled{1} \& \protect\circled{3} imply \protect\circled{4}. \protect\circled{4} implies \protect\circled{5}. \protect\circled{3} \& \protect\circled{5} imply \protect\circled{6}.}
    \label{fig:claim_1_proof}
\end{figure}

%% file: figures/claim_2_1_case_2_proof.tex
\def\length{14}
\def\start{4}

\begin{figure}
    \centering
    \begin{tikzpicture}
        \node (a) at (0, 2) {$p(o_{\min})$};
    
        \draw[->] (a.east) -- (\length, 2);

        \node (a) at (\start, 3) {\cref{line:g_a_query}};

        \draw[-] (\start, 2) -- (a.south);

        \node (b) at (\start - 2, 3 + 1) {After $\max(T, D)$};

        \path[-latex, dashed] (b) edge (a);

        \node (b) at (1/2*\length, 3) {\cref{line:help}};

        \draw[-] (1/2*\length, 2) -- (b.south);

        \path[-latex, dashed] (a) edge [bend left] node[above] {\Cref{def:stuck}} (b);

        \node (c) at (\length - \start, 3) {\cref{line:g_a_cas}};

        \draw[-] (\length - \start, 2) -- (c.south);

        \path[-latex, dashed] (b) edge [bend left] node[above] {\Cref{lemma:if_observed_done_wont_try_to_do}} (c);

        \node (b) at (\length - \start + 2, 3 + 1) {Implies $\announceobject{}$ changes after $T$};

        \path[-latex, dashed] (b) edge (c);
    \end{tikzpicture}
    \caption{Run constructed in case 2 in the proof of \Cref{claim:all_in_o_get_out}. Since the operation $o \in O_<$ is done, $p(o_{\min})$ tries to store $o_{\min}$ into $\announceobject{}$. Whether it succeeds or not, $o$ is no longer in $\announceobject{}$ after $T$.}
    \label{fig:claim_2_1_case_1_proof}
\end{figure}
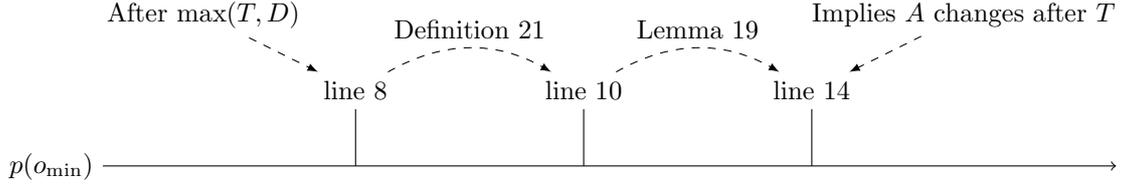

%% file: figures/claim_2_proof.tex
\def\length{14}
\def\start{3.325}

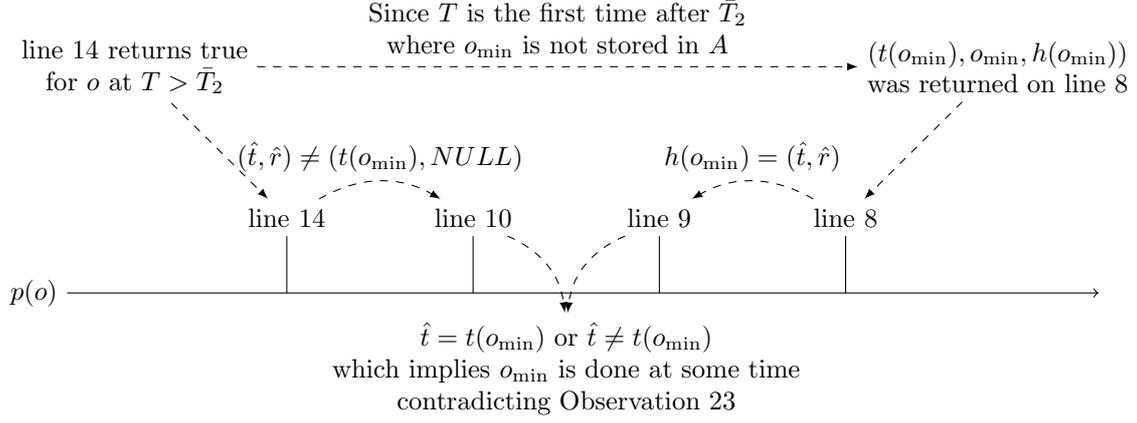
\begin{figure}
    \centering
    \begin{tikzpicture}
        \node (a) at (0, 2) {$p(o)$};
    
        \draw[->] (a.east) -- (\length, 2);

        \node (a) at (\start, 3) {\cref{line:g_a_cas}};

        \draw[-] (\start, 2) -- (a.south);

        \node (b) at ({\start + 1/3*(\length - 2*\start)}, 3) {\cref{line:help}};

        \draw[-] ({\start + 1/3*(\length - 2*\start)}, 2) -- (b.south);

        \node (c) at ({\start + 2/3*(\length - 2*\start)}, 3) {\cref{line:help_pointer_query}};

        \draw[-] ({\start + 2/3*(\length - 2*\start)}, 2) -- (c.south);

        \node (d) at ({\start + 3/3*(\length - 2*\start)}, 3) {\cref{line:g_a_query}};

        \draw[-] ({\start + 3/3*(\length - 2*\start)}, 2) -- (d.south);

        \node[align=center] (e) at (\start - 2, 5) {\cref{line:g_a_cas} returns true\\ for $o$ at $T > \T_2$};

        \path[-latex, dashed] (e) edge (a);

        \node[align=center] (f) at ({\start + 3/3*(\length - 2*\start) + 2}, 5) {$(t(o_{\min}), o_{\min}, \helpobjectsymbol{}(o_{\min}))$\\ was returned on \cref{line:g_a_query}};

        \path[-latex, dashed] (f) edge (d);

        \path[-latex, dashed] (e) edge node[above, align=center] {Since $T$ is the first time after $\T_2$\\ where $o_{\min}$ is not stored in $\announceobject{}$} (f);

        \path[-latex, dashed] (d) edge [bend right] node[above, align=center] {$\helpobjectsymbol{}(o_{\min}) = (\hat{t}, \hat{r})$} (c);

        \path[-latex, dashed] (a) edge [bend left] node[above, align=center] {$(\hat{t}, \hat{r}) \neq (t(o_{\min}), NULL)$} (b);

        \node[align=center] (g) at (1/2*\length, 1) {$\hat{t} = t(o_{\min})$ or $\hat{t} \neq t(o_{\min})$\\ which implies $o_{\min}$ is done at some time\\ contradicting \Cref{observation:stuck_implies_not_done}};

        \path[-latex, dashed]
            (b) edge [bend left] (g)
            (c) edge [bend right] (g)
            ;
    \end{tikzpicture}
    \caption{Run in the proof of \Cref{claim:o_min_always} after concluding $o$ was stored in $\announceobject{}$ at $T > \T_2$ because \cref{line:g_a_cas} returned true for $o$. This is because if \cref{line:g_a_gcas} returned true for $o$ then $o \in O_<$ contradicting \Cref{claim:all_in_o_get_out}.}
\end{figure}

%% file: sections/proof_of_linearizability.tex
\section{Proof of Linearizability}

In this section, we prove that \Cref{alg:wait-free-simple} is linearizable.
We start with some intermediate lemmas.

\begin{lemma}\label{lemma:done_implies_help_struct_form}
    If an operation $o$ is done at time $T$ then for all times $T' \geq T$ $\helpobjectsymbol{}(o) \neq (t(o), NULL)$ at $T'$.
\end{lemma}

\begin{proof}
    Since $o$ is done at $T$ there exists a time $T^* \leq T$ where $\helpobjectsymbol{}(o) = (t(o), r)$ at $T^*$ such that $r \neq NULL$ by \Cref{def:done}.
    Since by \Cref{observation:only_owner_changes_times_of_help_struct} the time field of $\helpobjectsymbol{}(o)$ is monotonically increasing, at any time $T' \geq T^*$ it equals some timestamp $t \geq t(o)$.
    If $t = t(o)$ then by \Cref{lemma:once_response_is_not_null_only_owner_sets_to_null} $\helpobjectsymbol{}(o) = (t(o), r)$ at $T'$.
    Otherwise, $\helpobjectsymbol{}(o) = (t, -)$ with $t > t(o)$ at $T'$.
    Therefore, $\helpobjectsymbol{}(o) \neq (t(o), NULL)$ at $T'$.
    \qH{\Cref{lemma:done_implies_help_struct_form}}
\end{proof}

\begin{lemma}\label{lemma:double_execution_implies_concurrent}
    Suppose that:
    \begin{itemize}
        \item At time $T$, a process $p$ executes \cref{line:g_r_cas} for an operation $o$ and this execution returns true.
        \item At time $T' >T$, a process $p'$ executes \cref{line:g_r_cas} also for operation $o$.
        This occurs in an iteration of the loop in line 4.
    \end{itemize}
    Then $p'$ executed \cref{line:g_r_query} in this iteration at some time $T_1 < T$.
\end{lemma}
% \hline
\begin{proof}
    Suppose for a contradiction that $T < T_1$.
    Since $p'$ executed \cref{line:g_r_query} and \ref{line:g_r_cas} at times $T_1$ and $T'$, respectively, in the same iteration of the loop, $p'$ executed \cref{line:help_pointer_cas}, \ref{line:g_a_query}, \ref{line:help_pointer_query}, and \ref{line:help} at some times $T_2$, $T_3$, $T_4$, and $T_5$, respectively, such that $T_1 < T_2 < T_3 < T_4 < T_5 < T'$.
    Since $p'$'s executions of \cref{line:g_r_cas} at $T'$ was for $o$, $p'$ was returned $(t(o), o, \helpobjectsymbol{}(o))$ on \cref{line:g_a_query} at $T_3$.
    Hence, $p'$'s execution of \cref{line:help_pointer_query} at $T_4$ was for $\helpobjectsymbol{}(o)$.
    Since $p'$ executed \cref{line:g_r_cas} at $T'$, the condition of \cref{line:help} was true at $T_5$.
    Thus, $\helpobjectsymbol{}(o) = (t(o), NULL)$ at $T_4$.

    We now complete the proof by showing that $\helpobjectsymbol{}(o) \neq (t(o), NULL)$ at $T_4$.
    We do so by showing that $o$ is done at $T_4$.
    There are two cases.
    \begin{enumerate}
        \item[] \hspace{-15pt}\textbf{Case 1.} $o$ is stored in $\stateobject{}$ at $T_1$.
        
        Hence, $p'$'s execution of \cref{line:help_pointer_cas} at $T_2$ was for $o$.
        Thus, by \Cref{lemma:copy_help_struct_implies_done} $o$ is done at $T_2$.
        By \Cref{observation:if_once_done_then_always_done}, $o$ is done at $T_4$, as wanted.
        
        \item[] \hspace{-15pt}\textbf{Case 2.} $o$ is not stored in $\stateobject{}$ at $T_1$.

        Since \cref{line:g_r_cas} was executed for $o$ and returned true at $T$, $\stateobject{}$ stored $o$ at $T$.
        Since some $o' \neq o$ is stored in $\stateobject{}$ at time $T_1 > T$, by \Cref{lemma:non_null_help_struct_response} $o$ is done at $T_1$.
        By \Cref{observation:if_once_done_then_always_done}, $o$ is also done at $T_4 > T_1$ as wanted.
    \end{enumerate}
    Since $o$ is done at $T_4$, by \Cref{lemma:done_implies_help_struct_form} $\helpobjectsymbol{}(o) \neq (t(o), NULL)$ at $T_4$, a contradiction.
    \qH{\Cref{lemma:double_execution_implies_concurrent}}
\end{proof}

\newcommand{\Seq}{Seq}

\begin{lemma}\label{lemma:sequence_is_unique}
		Let $\Seq$ be the finite sequence of executions of \cref{line:g_r_cas} that returned true up to and including some time $T$.
    	Let $k$ be the length of $\Seq$ and $\Seq_i$ denote the $i$th element of $\Seq$.
		For all $i$ and $j$ such that $1 \leq i \leq k$ and $1 \leq j < i$, $\Seq_j$ is for a different operation than $\Seq_i$.
\end{lemma}

\begin{proof}
    We prove this by strong induction on $i$.
    Let:
    \begin{align*}
        P(i) = \text{for all $j \in [1, i)$ $\Seq_j$ is for a different operation than $\Seq_i$}
    \end{align*}
    \begin{enumerate}
        \item[] \hspace{-20pt}\textbf{BC.} $P(1)$.

        This is immediate because the interval $[1, 1)$ is empty.

        \item[] \hspace{-15pt}\textbf{IC.} $\forall k' \in [1, k)\ ([\forall i \in [1, k']\ P(i)] \implies P(k' + 1))$.

        Suppose that, for some $k' \in [1, k)$, $\forall i \in [1, k']\ P(i)$ holds: this is the inductive hypothesis.
        We will show this implies $P(k' + 1)$ holds.
        Suppose for a contradiction that $P(k' + 1)$ does not hold.
        Thus, there exists a $j^* \in [1, k' + 1)$ such that $\Seq_{j^*}$ is for the same operation as $\Seq_{k' + 1}$.
        Let $o$ be the operation that $\Seq_{j^*}$ and $\Seq_{k' + 1}$ executed \cref{line:g_r_cas} for.
        Let $p$ be the process that executed \cref{line:g_r_cas} at $\Seq_{k' + 1}$.
        Let the times that $\Seq_{j^*}$ and $\Seq_{k' + 1}$ are executed at be $T_{j^*}$ and $T_{k' + 1}$, respectively.
        Since $j^* < k' + 1$, by the definition of $\Seq$, $T_{j^*} < T_{k' + 1}$.
        In the same iteration of the loop that $p$ executed $\Seq_{k' + 1}$, let the time $p$ executed \cref{line:g_r_query} at be $T_1 < T_{k' + 1}$.
        Since  $T_{j^*} < T_{k' + 1}$ and both $\Seq_{j^*}$ and $\Seq_{k' + 1}$ are for $o$, by \Cref{lemma:double_execution_implies_concurrent}, $T_1 < T_{j^*}$.
        Let $o_1$ be the operation stored in $\stateobject{}$ at $T_1$.
        The remainder of the proof shows that there are two different executions of \cref{line:g_r_cas} for $o_1$ that return true before $\Seq_{k' + 1}$ in $\Seq$ contradicting the inductive hypothesis.
        We start by showing that $o_1 \neq o$.
        There are two cases.
        \begin{enumerate}
            \item[] \hspace{-15pt}\textbf{Case 1.} $o_1 = NOOP$.

            Since \cref{line:g_r_cas} was executed for $o$, $o$ was stored in $\announceobject{}$ at some time, and hence some process executed DoOp($o$).
            Hence, $o \neq NOOP$ by \Cref{assumption:1}.
            Therefore $o_1 \neq o$ as wanted.
            
            \item[] \hspace{-15pt}\textbf{Case 2.} $o_1 \neq NOOP$.

            Since $o_1$ is stored in $\stateobject{}$ at $T_1$, and $o_1 \neq NOOP$, some process must have executed \cref{line:g_r_cas} for $o_1$ that returned true before $T_1$.
            Since $T_1 < T_{j^*}$, by the definition of $\Seq$, this implies that
            $\Seq_{j}$ for some $j < j^*$ was for $o_1$.
            Since $j^* < k' + 1$ by the inductive hypothesis $P(j^*)$ holds.
            Hence, $\Seq_{j}$ is for a different operation than $\Seq_{j^*}$ since $j < j^*$.
            Therefore, $o_1 \neq o$ as wanted.
        \end{enumerate}

        We now return to the proof of the inductive case.
        Since $o_1$ is stored in $\stateobject{}$ at $T_1$, and $\Seq_{k' + 1}$ returns true at $T_{k' + 1}$, $o_1$ must be stored in $\stateobject{}$ at the time $\Seq_{k' + 1}$ is executed.
        Thus, since $o$ is stored in $\stateobject{}$ at $T_{j^*}$, $o_1 \neq o$, and $T_1 < T_{j^*} < T_{k' + 1}$, \cref{line:g_r_cas} must be executed for $o_1$ and return true some time in between $T_{j^*}$ and $T_{k' + 1}$.
        This implies: (1) by the definition of $\Seq$, $\Seq_{j'}$ for some $j^* < j' < k' + 1$ is for $o_1$ and (2) $o_1$ was stored in $\announceobject{}$ at some time, hence some process executed DoOp($o_1$), and thus $o_1 \neq NOOP$ by \Cref{assumption:1}.
        Since $o_1 \neq NOOP$ and $o_1$ is stored in $\stateobject{}$ at $T_1$, some process must have must have executed \cref{line:g_r_cas} for $o_1$ that returned true before $T_1$.
        Since $T_1 < T_{j^*}$ by the definition of $\Seq$ this implies that $\Seq_{j}$ for some $j < j^*$ was for $o_1$.
        Hence, both $\Seq_{j}$ and $\Seq_{j'}$ are for the same operation $o_1$.
        Since $j < j' < k + 1$ this contradicts the inductive hypothesis that $P(j')$ holds.
    \end{enumerate}
    Therefore by strong induction, we have $\forall i \in [1, k]\ P(i)$ holds as wanted.
    \qH{\Cref{lemma:sequence_is_unique}}
\end{proof}

\begin{lemma}\label{lemma:linearized_once_implies_never_again}
    If \cref{line:g_r_cas} is executed for an operation $o$ and it returns true at some time $T$, then all executions of \cref{line:g_r_cas} for $o$ after time $T$ return false.
\end{lemma}

\begin{proof}
    Suppose for a contradiction that \cref{line:g_r_cas} is executed twice for an operation $o$ and it returns true both at times $T$ and $T'$ such that $T < T'$.
    Consider the sequence $\Seq$ constructed in \Cref{lemma:sequence_is_unique} up to $T'$.
    Recall that $\Seq$ is the sequence of all executions of \cref{line:g_r_cas} that return true before or at time $T'$.
    Since, by assumption \cref{line:g_r_cas} is executed twice for $o$ at times $T$ and $T'$ such that $T < T'$, both of these executions for $o$ are in $\Seq$; let them be $\Seq_i$ and $\Seq_j$, respectively.
    \Cref{lemma:sequence_is_unique} asserts that $\Seq_i$ and $\Seq_j$ are for different operations.
    However, $\Seq_i$ and $\Seq_j$ are both for $o$, a contradiction.
    \qH{\Cref{lemma:linearized_once_implies_never_again}}
\end{proof}

\Cref{lemma:linearized_once_implies_never_again} immediately implies:

\begin{corollary}\label{corollary:linearized_at_most_once}
    For all operations $o$, \cref{line:g_r_cas} is executed for $o$ and returns true at most once.
\end{corollary}

\begin{lemma}\label{lemma:complete_implies_linearized}
    If, during the execution of DoOp($o$), process $p(o)$ executes \cref{line:op_done} at some time $T$  then \cref{line:g_r_cas} was executed for $o$ and it returned true at some time $T' < T$.
\end{lemma}

\begin{proof}
    Suppose $p(o)$ executes \cref{line:op_done} at some time $T$ during the execution of DoOp($o$).
    This implies at some time $T_1 < T$, $p(o)$ executed \cref{line:loop_start} and the condition was false.
    Hence, $\helpobjectsymbol{}(o) \neq (t(o), NULL)$ at $T_1$.
    This implies that \cref{line:help_pointer_cas} was executed at some time $T_2 < T_1$ in the form of $\text{GCAS}(=, \helpobjectsymbol{}(o), (t(o), NULL), (t(o), r))$ where $r \neq NULL$ by \Cref{observation:help_pointer_cas_is_well_formed}.
    Thus, $o$ was stored in $\stateobject{}$ at some time $T_3 < T_2$.
    Therefore, \cref{line:g_r_cas} was executed for $o$ and it returned true at some time $T_4 < T_3$ which is before $T$ by transitivity.
    \qH{\Cref{lemma:complete_implies_linearized}}
\end{proof}

\Cref{lemma:linearized_once_implies_never_again} and \Cref{lemma:complete_implies_linearized} immediately imply:

\begin{corollary}\label{corollary:complete_implies_linearized_exactly_once}
    If operation $o$ is complete then there is exactly one execution of \cref{line:g_r_cas} for $o$ that returns true.
\end{corollary}

\begin{lemma}\label{lemma:response_is_from_linearization}
    Suppose operation $o$ is complete.
    By \Cref{corollary:complete_implies_linearized_exactly_once} there is exactly one execution of \cref{line:g_r_cas} for $o$ that returns true.
    By \Cref{observation:g_r_is_well_formed} this execution is of the form $\text{GCAS}(=, \stateobject{}, -, (t(o), -, r, \helpobjectsymbol{}(o)))$.
    Then $o$'s response is $r$.
\end{lemma}

\begin{proof}
    Suppose for a contradiction that $o$'s response is $r' \neq r$.
    Hence, $p(o)$ returned $r'$ on \cref{line:op_done}.
    This implies $\helpobjectsymbol{}(o) = (t(o), r')$.
    Hence, some process executed \cref{line:help_pointer_cas} in the form $\text{GCAS}(=, \helpobjectsymbol{}(o), (t(o), NULL), (t(o), r'))$ that returned true.
    Thus, this process was returned $(t(o), -, r', \helpobjectsymbol{}(o))$ on \cref{line:g_r_query}.
    This implies \cref{line:g_r_cas} was executed for $o$ in the form $\text{GCAS}(=, \stateobject{}, -, (t(o), -, r', \helpobjectsymbol{}(o)))$ and it returned true.
    However, the single execution of \cref{line:g_r_cas} for $o$ that returned true is assumed to be of the form $\text{GCAS}(=, \stateobject{}, -, (t(o), -, r, \helpobjectsymbol{}(o)))$.
    This implies \cref{line:g_r_cas} was executed for $o$ and returned true twice, a contradiction to \Cref{corollary:complete_implies_linearized_exactly_once}.
    \qH{\Cref{lemma:response_is_from_linearization}}
\end{proof}

\begin{theorem}\label{theorem:linearizable}
    \Cref{alg:wait-free-simple} is linearizable.
\end{theorem}

\begin{proofsketch}
    Recall from \Cref{sec:model} to prove this theorem we need to show that any history $H$ of \Cref{alg:wait-free-simple} has a completion $H'$ and a linearization $L$ of $H'$ that conforms to type $\mathcal{T}$.
    We do so in three steps.
    
    We first construct a completion $H'$ of $H$ as follows.
    For each incomplete operation $o$ in $H$: if \cref{line:g_r_cas} was executed for $o$ and it returns true, then it is completed in $H'$ by returning the response written into $\stateobject{}$ by the execution of \cref{line:g_r_cas} for $o$ that returned true (by \Cref{corollary:linearized_at_most_once} this response is well-defined), otherwise, $o$ is removed from $H'$.
    
    We then construct a linearization $L$ of $H'$ as follows.
    Each operation $o$ in $H'$ is linearized when \cref{line:g_r_cas} is executed for $o$ and it returns true.
    By \Cref{corollary:linearized_at_most_once} and \Cref{corollary:complete_implies_linearized_exactly_once} for each operation $o$ in $H'$ this occurs exactly once.
    Hence, the linearization point of each operation in $H'$ is well defined.
    We now justify that these linearization points occur within the execution interval of their respective operations.
    Consider an arbitrary operation $o$ in $H'$.
    Since \cref{line:g_r_cas} is only executed for operations if they are stored in $\announceobject{}$, and operations are only stored in $\announceobject{}$ after they are invoked, the linearization point of $o$ occurs after it is invoked.
    Since $o$ completes $p(o)$ executed \cref{line:op_done} at $T$.
    Hence, by \Cref{lemma:complete_implies_linearized} its linearization point is at some time before $T$.
    Therefore, $o$'s linearization point occurs within its execution interval.

    We now show that $L$ conforms to type $\mathcal{T}$.
    Let the state of an object of type $\mathcal{T}$ be the state field stored in $\stateobject{}$.
    Since the linearization point of an operation $o$ in $L$ also updates the state of an object of type $\mathcal{T}$, the linearization points are simultaneous with all the state changes.
    It's left to show that these state changes and returned responses respect the state transition relation $\delta$ of $\mathcal{T}$.
    Suppose that the state changes to $s'$ after linearizing $o$ and the response is $r'$.
    Observe that $s'$ and $r'$ are determined by the previous state $s^*$ read on \cref{line:g_r_query} followed by applying $o$ to $s^*$ on \cref{line:apply} using $apply_{\mathcal{T}}$.
    Hence, $(s^*, o, s', r') \in \delta$ (the state transition relation of type $\mathcal{T}$).
    Finally, observe that $o$ returns $r'$ because of the following.
    If $o$ is complete in $H$, by \Cref{lemma:response_is_from_linearization} $o$'s response is equal to the response stored in $\stateobject{}$ at the time $o$ is linearized (which is $r'$ in this case).
    Otherwise, if $o$ is incomplete in $H$, by the construction of $H'$ $o$'s response is the same.
    \qH{\Cref{theorem:linearizable}}
\end{proofsketch}

%% file: sections/conclusion.tex
\section{Conclusion}

In this paper, we first proposed a natural generalization of the well-known \CAS{} object, one that replaces the equality comparison with an arbitrary comparator.
We then illustrated the use of GCAS by presenting a simple wait-free universal construction for an infinite arrival distributed system and proved its correctness.
In this universal construction, the value stored in each GCAS object has at most four fields.
Since current hardware only supports CAS objects that can store up to two fields\footnote{These objects are called \emph{Double Width CAS}. Modern Intel 64-bit processors provide this object through the \texttt{cmpxchg16b} instruction~\cite{intel}.}, the reader may wonder whether we can modify our universal construction so that each GCAS object also stores at most two fields.
The answer is yes, as we will describe in a future version of this paper.

%% file: main.bbl
\begin{thebibliography}{10}

\bibitem{fatourou2009redblue}
P.~Fatourou and N.~D. Kallimanis, ``The redblue adaptive universal constructions,'' in {\em International Symposium on Distributed Computing}, pp.~127--141, Springer, 2009.

\bibitem{herlihy1991wait}
M.~Herlihy, ``Wait-free synchronization,'' {\em ACM Transactions on Programming Languages and Systems (TOPLAS)}, vol.~13, no.~1, pp.~124--149, 1991.

\bibitem{fatourou2011highly}
P.~Fatourou and N.~D. Kallimanis, ``A highly-efficient wait-free universal construction,'' in {\em Proceedings of the twenty-third annual ACM symposium on Parallelism in algorithms and architectures}, pp.~325--334, 2011.

\bibitem{anderson1995universal}
J.~H. Anderson and M.~Moir, ``Universal constructions for multi-object operations,'' in {\em Proceedings of the fourteenth annual ACM symposium on Principles of distributed computing}, pp.~184--193, 1995.

\bibitem{chuong2010universal}
P.~Chuong, F.~Ellen, and V.~Ramachandran, ``A universal construction for wait-free transaction friendly data structures,'' in {\em Proceedings of the twenty-second annual ACM symposium on Parallelism in algorithms and architectures}, pp.~335--344, 2010.

\bibitem{afek1995wait}
Y.~Afek, D.~Dauber, and D.~Touitou, ``Wait-free made fast,'' in {\em Proceedings of the twenty-seventh annual ACM symposium on Theory of computing}, pp.~538--547, 1995.

\bibitem{naderibeni2023wait}
H.~Naderibeni and E.~Ruppert, ``A wait-free queue with polylogarithmic step complexity,'' in {\em Proceedings of the 2023 ACM Symposium on Principles of Distributed Computing}, pp.~124--134, 2023.

\bibitem{michael1996simple}
M.~M. Michael and M.~L. Scott, ``Simple, fast, and practical non-blocking and blocking concurrent queue algorithms,'' in {\em Proceedings of the fifteenth annual ACM symposium on Principles of distributed computing}, pp.~267--275, 1996.

\bibitem{merritt2000computing}
M.~Merritt and G.~Taubenfeld, ``Computing with infinitely many processes: Under assumptions on concurrency and participation,'' in {\em International Symposium on Distributed Computing}, pp.~164--178, Springer, 2000.

\bibitem{herlihy1990linearizability}
M.~P. Herlihy and J.~M. Wing, ``Linearizability: A correctness condition for concurrent objects,'' {\em ACM Transactions on Programming Languages and Systems (TOPLAS)}, vol.~12, no.~3, pp.~463--492, 1990.

\bibitem{intel}
Intel, {\em Intel 64 and IA-32 Architectures Software Developer’s Manual, Volume 2 (2A, 2B, 2C, \& 2D): Instruction Set Reference, A-Z}, 2024.

\end{thebibliography}
